\documentclass{eptcs}
 \providecommand{\firstpage}{27}

\usepackage{breakurl}
\usepackage[strings]{underscore}
\usepackage[english]{babel}
\usepackage{color}
\usepackage{MnSymbol}
\usepackage{bussproofs}
\usepackage{tikz}
\usepackage{tkz-graph}
\usepackage{lscape}
\usepackage{amsthm}
\usepackage{rotating}
\usetikzlibrary{arrows,shapes}
\usepackage{ifthen}
\usepackage{calc}
\usepackage{multido}

\newcommand{\tuple}[1]{{\langle #1 \rangle}}
\newtheorem{definition}{Definition}
\newtheorem{theorem}{Theorem}

\newtheorem{proposition}{Proposition}
\newcommand{\partset}[1]{\ensuremath{\textbf{2}^{#1}}}
\newcommand{\eqdef}[0]{\; \triangleq \; } 
\newcommand{\abbrev}[1]{#1, \relax}
\newcommand{\ie}[0]{\abbrev{\textit{i.e.}}}
\newcommand{\eg}[0]{\abbrev{\textit{e.g.}}}

\newcommand{\sem}[1]{\left \lsem {#1} \right \rsem}
\newcommand{\A}[1]{[{#1}]}
\newcommand{\E}[1]{\langle{#1}\rangle}
\newcommand{\hybridset}[0]{\textsf{F}_{\mathcal H}}
\newcommand{\restrict}[2]{{#1}|_{#2}}
\newcommand{\dom}[1]{\text{Dom} \; #1} 
\newcommand{\cause}[0]{\mathrel{\medcircle\kern-3.25pt \to}} 
\newcommand{\pcause}[0]{\mathrel{\odot\kern-3.25pt \to}} 
\newcommand{\remcause}[0]{\mathrel{\oplus\kern-3.25pt \to}} 
\newcommand{\allcause}[0]{\mathrel{\ostar\kern-3.25pt \to}} 
\newcommand{\reg}[0]{\longrightarrow}
\newcommand{\activate}[0]{\stackrel{+}{\reg}}
\newcommand{\inhibit}[0]{\stackrel{-}{\reg}}
\newcommand{\htrace}[1]{
\setcounter{ti}{\theti+\thetinterv}
\draw (\theti,\thehauteur+0.2) -- (\theti,-0.1) node[below] {\tiny $#1 \strut$}; }
\newcommand{\htracei}[2]{
\setcounter{ti}{\theti+\thetinterv}
\pgfmathparse{\theti+#1-1}\let\z\pgfmathresult
\foreach \x in {\theti,...,\z} {
	\draw (\x,\thehauteur+0.2) -- (\x,-0.1) node[below] {\tiny $#2 \strut$}; }
\setcounter{ti}{\theti+#1-1} }

\newcommand{\hist}[2]{
\pgfmathparse{\thetimeg+(#1-\thetimeg)/2}\let\z\pgfmathresult
\setcounter{timeg}{#1}
\draw [very thick](\thetimeg,\thehauteur+0.2) -- (\thetimeg,-0.1);
\draw [very thick](\thetimeg,-0.6) -- (\thetimeg,-1.3);
\node at (\z,-1){\small $#2 \strut$};
}
\newcommand{\hcmax}[2]{
\draw [dashed] (-0.1,#1) -- (11,#1);
\node at (-0.3,#1 ) {\textbf{#2}};
}
\newcounter{timeg}
\newcounter{hauteur}
\newcounter{longueur}
\newcounter{tinterv}
\newcounter{ti}
\newenvironment{hgraph}[3]{
\setcounter{longueur}{#1}
\setcounter{hauteur}{#2}
\setcounter{tinterv}{#3}
\setcounter{timeg}{0}
\setcounter{ti}{0}
\begin{tikzpicture}[xscale=0.6, yscale=0.7]
\draw[->,very thick] (0,0) -- (#1+0.3,0);
\node at (#1-0.5,0.5){\footnotesize \textit{time}};
\draw[->,very thick] (0,-1.3) -- (0,#2+0.3);
\node at (0.4,#2){\footnotesize \textit{q.}};
\node at (-0.27,-0.4) {\small T};
\draw [semithick](0,-0.7) -- (#1,-0.7);
\node at (-0.27,-1) {\small H};
\draw [semithick](0,-1.3) -- (#1,-1.3);
}%
{\end{tikzpicture}}

\title{ \textsc{gubs}, a Behavior-based Language for Open System \\
Dedicated to Synthetic Biology}
\author{Adrien Basso-Blandin
\institute{IBISC lab. }
 \institute{Evry University}
\email{abasso@ibisc.univ-evry.fr}
\and 
 Franck Delaplace
\institute{IBISC lab.}
 \institute{Evry University}
\email{franck.delaplace@ibisc.univ-evry.fr}}

\def\titlerunning{\textsc{gubs}, a Behavioral Language for Synthetic Biology}
\def\authorrunning{A.~Basso-Blandin and F.~Delaplace}

\begin{document}
\maketitle

\begin{abstract}
In this article, we propose a domain specific language, GUBS (Genomic Unified Behavior Specification), dedicated to the behavioural specification of synthetic biological devices, viewed as discrete open dynamical systems. \textsc{gubs} is a rule-based declarative language. By contrast to a closed system, a program is always a partial description of the behaviour of the system. The semantics of the language accounts the existence of some hidden non-specified actions that possibly alter the behaviour of the programmed devices.
The compilation framework follows a scheme similar to automatic theorem proving, aiming at improving synthetic biological design safety.

\end{abstract}
\section{Introduction}

Synthetic biology is an emerging scientific field combining the investigative nature of biology with the constructive nature of engineering~\cite{Purnick2009} to design synthetic biological systems. The issue is to devise new functionality/behaviour that does not exist in nature.
Then, the field of synthetic biology is looking forward principles and tools to make the biological devices inter-operable and programmable~\cite{Lu2009}. Synthetic biology projects were first focusing on the design and the improvement of small genetic devices comparable to logical gates for electronic circuits~\cite{Regot2010,Clancy2010}. Recently, projects have attempted to develop large bio-systems integrating different devices with as a long-term goal, the design of \textit{de-novo} synthetic genome~\cite{Gibson2010}. In this endeavour, the computer-aided-design (CAD) environments play a central role by providing the required features to engineer systems: specification, analysis, and tuning~\cite{Bilitchenko2011,Pedersen2009a, Umesh2010, Czar2009}. Pioneer applications show the valuable potential of such environments in \textsc{igem} competition. 


Currently, the design specifies the structural assembly of \textsc{dna} sequences (biobrick) as in \textsc{genocad}~\cite{Cai2009}. Although this description is indispensable to provide a finalized specification of devices, the abstraction level seems inappropriate for tackling large bio-systems. The required size of programs for sequence description likely makes the task error-prone and un-come-at-able. In the same way as large softwares cannot be programmed in binary, large biological systems cannot be described as a\textsc{dna} sequence assembly. Then, scaling up the complexity of the synthetic biological systems needs to complete the structural description by an additional abstract programming layout based on a high-level programming language and harness the automatic conversion of the design specification into a DNA sequence, like compilers. High level programming language for synthetic biology is announced as a key milestone for the second wave of synthetic biology to overcome the complexity of large synthetic system design~\cite{Purnick2009}. Nonetheless, in this domain, language technology is still in its infancy and transforming this vision into a concrete reality remains a daunting challenge. 

Such high-level language should describe the devices in term of functionalities, offering the ability to program the behaviour directly instead of the structure supporting this behaviour.
Indeed, behaviour specification contributes to accurately document the device by adding its behavioural description, to assess its functionality automatically and formally, notably by generating test-benches from this specification, and to get a relative independence to technology because different biological structures can carry out the same functionality. In this framework, the components are selected and organized automatically or semi-automatically to generate a structural description of the device at compile phase whose behaviour complies with the specified function.
A such approach has been already achieved in hardware by using languages as \textsc{vhdl}~\cite{Ashenden2002} or \textsc{verilog}~\cite{Thomas1998} to overcome the growing complexity of electronic circuits. However, the major difference in synthetic biology relates to the openness of biological system. Thereby, the issue is to propose a behavioural language for open systems.
More precisely, \textsc{gubs} is a rule-based declarative language dedicated to the behavioural specification of \emph{discrete open dynamical systems }for synthetic biology interacting with its environment. \textsc{gubs} symbolically defines
the behaviours to provide a relative independence from structures by postponing the biological component selection at compile phase. Within this framework, the compiler translates the behavioural specification to a structural description of a device whose behaviour carries the functional features defined by a program. The proposed compilation method is inspired by automated theorem proving. 

After introducing the main features of \textsc{gubs} language (Section~\ref{sec:gubs}), we define the semantics of \textsc{gubs} based on hybrid logic. Then, we detail the proof-based principles governing the compilation (Section~\ref{sec:compil}) illustrated with a complete example (Section~\ref{sec:example}). After a survey of the related works, Section~\ref{sec:related}, we conclude (Section~\ref{sec:conclusion}).

\section{\textsc{gubs} language}
\label{sec:gubs}
\newcommand{\context}[1]{[#1]}

In this section, we describe the main features of \textsc{gubs}. 
\paragraph{Constant and variables.} In \textsc{gubs}, two kinds of objects are distinguished: the \emph{constants} and the \emph{variables}. The constants designate the pre-defined objects in a corpus of knowledge. In biology, the constants may refer to proteins or genes of interest. For example, the agent $\textit{LacZ}$ refers to LacZ protein or gene. By convention, their name starts with a capital letter. The variables refer to an abstraction of these pre-defined objects and can be potentially replaced (substituted) by any constant. By convention, the variable names start with a minuscule letter.
\paragraph{Agents, attributes and states.}
The \emph{agents} represent the biological objects. Their different observable \emph{states} characterize their different \emph{behaviours}. The behaviours actually define the different capacities for actions on the state of the other agents. They are characterized symbolically by a set of \emph{attributes} categorizing these different capacities. The real significance of the attributes is a matter of convention depending on the targeted realization (\eg protein pathways, gene network) and will be addressed through examples. 
For instance, the regulatory activity of a gene is observationally related to thresholds of \textsc{rna} transcripts concentration. At a given threshold, a gene regulates a given set of genes whereas at another one the regulation applies to another set of genes (See Figure~\ref{fig:causes}). The different thresholds define the levels of gene activities leading to different regulatory activities.
For a gene $G$, if we identify three different kinds of regulatory activities, the state of this gene will be defined by three different attributes $\{\textit{Low}, \textit{Mid}, \textit{High} \}$ that characterize symbolically three possible behaviours. 
For example, $G(\textit{Low})$ expresses the fact that agent $G$ is in state $\textit{Low}$ and then ready for the action corresponding to this attribute. In some cases, a single state is sufficient to qualify the capacity for the action of the agent. Hence, the agent is identified to its capacity. Then, $G$ means that agent $G$ is available. 

By contrast, $G(\overline{\textit{Low}})$ signifies that the state of the agent differs from $\mathit{Low}$ ($\overline G$ when an agent has a single capacity). It is worth to point out that, not being in a state defined by an attribute, does not necessarily means that the agent state is in another attribute. Indeed, for open systems the state of the agents could be of any sort that does not necessarily belong to the pre-defined attributes. 

 Two kinds of relations on attributes are defined: an order, $\prec$, meaning ``less capacity than'' and an inequality, $\napprox$, meaning ``different capacity than''. Then $\mathit{Low} \prec \mathit{Mid}$ implies that the capacity for the action of $\mathit{Mid}$ includes the capacity related to $ \mathit{Low}$. Usually, in gene regulatory model~\cite{Delaplace2010}, the set of genes regulated at a given level will also be regulated at a higher level. By contrast, in signalling pathways, the phosphorylation of a protein induces a conformational change of the structure leading to a specific signalling potentiality not occurring in the unphosphosrylated conformation. Assuming that $\textit{Phos}$ and $\textit{UnPhos}$ respectively represents the phosphorylated and the unphosphorylated conformations of protein $P$, we have $\textit{Phos} \napprox \textit{UnPhos}$. Then, $P(\textit{Phos})$ implies $P(\overline{\textit{UnPhos}})$ implicitly. The attributes and the relation between attributes will be declared as follows:
$
 G :: \{\textit{Low} \prec \textit{Mid}, \textit{Mid} \prec \textit{High} \},
 P :: \{ \textit{Phos} \napprox \textit{UnPhos} \}.$
A simple set of attributes replaces the relations if unknown and no specific relation is set between attributes.

Finally, the description of the agent state is extended to a collection of agent states as follows: $g_1 + \ldots + g_n$, meaning that all the agent states, $g_i$, are observed concomitantly.

\paragraph{Trace, event, and history.} 
A \textsc{gubs} program describes a behaviour, its interpretation is based on the observations of designed systems. Then, the issue is to formalize the notion of behaviour observation. To this end, we focus on the notion of \emph{trace} that symbolically represents the evolution of some quantities related to the agents of interest by the evolution of these agent states. A trace can be obtained from experiments by establishing a correspondence between measurements of some quantities (\eg \textsc{rna} transcript concentration) and attributes of agents. Formally, a trace, $(T_t)_{1 \leq t \leq m}$, is a finite sequence of agent state sets where each set contains the agent states at a given instant. For instance, the evolution of a concentration evolving from $\textit{Low}$ to $\textit{High}$ for $G$ may be described by the following trace of $6$ instants: 
$
\renewcommand{\arraystretch}{0.5}
\scriptsize
\setlength{\doublerulesep}{\arrayrulewidth}
 \begin{array}[t]{ c @{} c @{} c @{} c @{} c @{} c @{} c}
(\{ G(Low) \}, & \{G(Low) \}, &\{G(Mid)\}, &\{G(Mid)\}, & \{G(Mid)\},& \{G(High)\}),& \hspace{4ex}. \\
\scriptscriptstyle 1 & \scriptscriptstyle 2 & \scriptscriptstyle 3 & \scriptscriptstyle 4 & \scriptscriptstyle 5 & \scriptscriptstyle 6 & \scriptscriptstyle 7
\end{array}$
However, all the events in a trace are not necessarily relevant with regards to the behaviour description. For example, if we focus on the evolution from $\mathit{Low}$ to $\mathit{High}$ for $G$ , only three events are relevant for the behaviour description: $G(\mathit{Low}), G(\mathit{Mid}), G(\mathit{High})$; without accounting the intermediary evolution stages occurring between. Then, the behaviour recognition always emphasizes the key events in a trace entailing its contraction to show their succession. Such a contracted series is called a \emph{consistent history} of the expected behaviour. Generally speaking, an history is related to a \emph{chronological division} of a trace into periods where the events of a period represent all the agent states occurring at each instant.
Then, an history is a sequence of these event sets.
Given a trace $(T_t)_{0 \leq t \leq m}$, and a chronological division, $(d_i)_{1\leq i \leq n}$, corresponding to a sequence of the starting dates for each period, the history is a sequence of agent states occurring in each period, $(H_i)_{1 \leq i < n}$, such that each $H_i=\bigcup_{d_i \leq t < d_{i+1}} T_t$. Hence, a consistent history is purposely made to point the characteristic event steps of a behaviour description out.

 In the previous example, a chronological division\footnote{Step $7$ is inserted as an extra step to comply with the definition of the chronological division.} of the trace leading to an history consistent with the expected evolution from $\textit{Low}$ to $\textit{High}$ for $G$ is $(1,3,6,7)$ which corresponds to following discrete time-intervals $([1,2],[3,5],[6,6]$).
The resulting history is: $\scriptstyle (\{ G(Low) \},\{G(Mid)\},\{G(High)\})$. Notice that $(1,2,4,7)$ also fits. However, the chronological division $(1,3,7)$ leads to an inconsistent history because the level \textit{Mid} is not seen as an intermediary event in the history (See also Figure~\ref{fig:causes} depicting the trace and consistent history of the dependences). The formal definition of the consistency in the scope of the semantics will be given in Section~\ref{sec:semantics}.

\paragraph{Behavioral dependence and observation spot.}
\begin{figure}[ht]
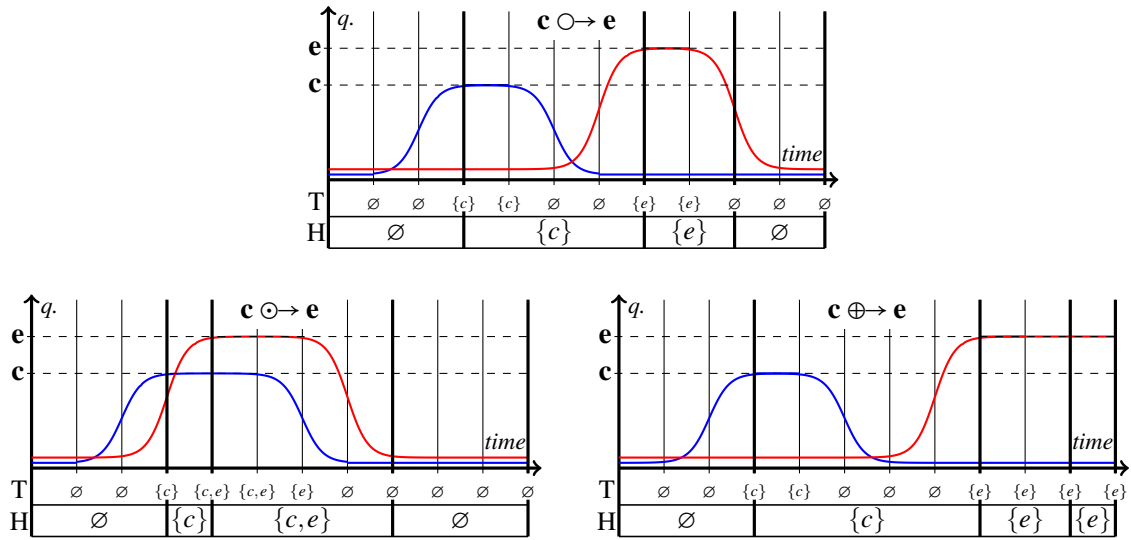

\begin{center}
\begin{hgraph}{11}{3}{1}
\htracei{2}{\emptyset}
\htracei{2}{\{c\}}
\htracei{2}{\emptyset}
\htracei{2}{\{e\}}
\htracei{3}{\emptyset}
\draw(5.5,3)node[fill=white] {$\textbf{c}\cause \textbf{e}$};
\begin{scope}
\clip (0,0) rectangle (11,3.2);
\draw [color=blue, thick](0,0.1) -- (1,0.1);
\draw[domain=-1:2,samples=100,color=blue,thick] plot ({\x+2},{1.7/(1+exp(-(\x*4.5)))+0.1});
\draw[domain=-2:1,samples=100,color=blue,thick] plot ({\x+5},{1.7/(1+exp(-(-(\x)*4.5)))+0.1});
\draw [color=blue,thick](6,0.1) -- (11,0.1);
\draw [color=red,thick](0,0.2) -- (4,0.2);
\draw[domain=-2:2,samples=100,color=red,thick] plot ({\x+6},{2.3/(1+exp(-(\x*4.5)))+0.2});
\draw[domain=-2:2,samples=100,color=red,thick] plot ({\x+9},{2.3/(1+exp(-(-(\x)*4.5)))+0.2});
\end{scope}
\hcmax{1.8}{c}
\hcmax{2.5}{e}
\hist{3}{\emptyset}
\hist{7}{\{c\}}
\hist{9}{\{e\}}
\hist{11}{\emptyset}
\end{hgraph}

\bigskip
\begin{hgraph}{11}{3}{1}
\htracei{2}{\emptyset}
\htrace{\{c\}}
\htracei{2}{\{c,e\}}
\htrace{\{e\}}
\htracei{5}{\emptyset}
\draw(5.5,3)node[fill=white] {$\textbf{c}\pcause \textbf{e}$};
\begin{scope}
\clip (0,0) rectangle (11,3.2);
\draw [color=blue,thick](0,0.1) -- (1,0.1);
\draw[domain=-1:2,samples=100,color=blue,thick] plot ({\x+2},{1.7/(1+exp(-(\x*4.5)))+0.1});
\draw[domain=-2:1,samples=100,color=blue,thick] plot ({\x+6},{1.7/(1+exp(-(-(\x)*4.5)))+0.1});
\draw [color=blue,thick](7,0.1) -- (11,0.1);
\draw [color=red,thick](0,0.2) -- (1,0.2);
\draw[domain=-2:2,samples=100,color=red,thick] plot ({\x+3},{2.3/(1+exp(-(\x*4.5)))+0.2});
\draw[domain=-2:2,samples=100,color=red,thick] plot ({\x+7},{2.3/(1+exp(-(-(\x)*4.5)))+0.2});
\draw [color=red,thick](9,0.2) -- (11,0.2);
\end{scope}
\hcmax{1.8}{c}
\hcmax{2.5}{e}
\hist{3}{\emptyset}
\hist{4}{\{c\}}
\hist{8}{\{c,e\}}
\hist{11}{\emptyset}
\end{hgraph}
\begin{hgraph}{11}{3}{1}
\htracei{2}{\emptyset}
\htracei{2}{\{c\}}
\htracei{3}{\emptyset}
\htracei{4}{\{e\}}
\draw(5.5,3)node[fill=white]{$\textbf{c}\remcause \textbf{e}$};
\begin{scope}
\clip (0,0) rectangle (11,3.2);
\draw[domain=-2:2,samples=100,color=blue, thick] plot ({\x+2},{1.7/(1+exp(-(\x*4.5)))+0.1});
\draw[domain=-2:2,samples=100,color=blue, thick] plot ({\x+5},{1.7/(1+exp(-(-(\x)*4.5)))+0.1});
\draw [color=blue,thick](7,0.1) -- (11,0.1);
\draw [color=red,thick](0,0.2) -- (5,0.2);
\draw[domain=0:4,samples=100,color=red,thick] plot ({\x+5},{2.3/(1+exp(-((\x-2)*4.5)))+0.2});
\draw [color=red,thick](8,2.5) -- (11,2.5);
\end{scope}
\hcmax{1.8}{c}
\hcmax{2.5}{e}
\hist{3}{\emptyset}
\hist{8}{\{c\}}
\hist{10}{\{e\}}
\hist{11}{\{e\}}
\end{hgraph}
\end{center}
\caption{The curves represent the typical behaviours of the causal dependences based on the time evolution of a quantity ($q$) related to agents $c$ and $e$ (\eg \textsc{rna} transcript for gene regulation). The symbolic agent states $c$ and $e$ are here both associated to the maximal threshold of the quantity. The symbolic trace ($T$) is issued from a periodic sampling of the evolution by identifying whether $c$ or $e$ occur. A consistent history ($H$) complying to a causal dependence definition is represented below the trace.
The first graphic illustrates the normal causality: $c\cause e$, the second the persistent: $c\pcause e$ and the third the remanent one: $c\remcause e$.}
\label{fig:causes}
\end{figure}
 A behavioural dependence identifies a relation between behaviours as a causal relation on events. Basically, the dependences should define the control of agents on another. However, the definition of the causality also needs to tackle the openness of a system by adapting it to this context. A seminal definition of the causality, proposed by Hume~\cite{Hume1739}, is formulated in terms of regularity on events: ``[we may define] a cause to be an object, followed by another, and where all the objects similar to the first are followed by objects similar to the second''. Although this definition appropriately characterizes the notion of control, the openness of the system implies to account the environment actions that possibly alter the causal dependence chain. For example, a programmed activation $G_1 \activate G_2$ may be contradicted by an existing inhibition $G_3 \inhibit G_2$ addressing the same target gene $G_2$. Hence, while $G_1$ is active, it may appear that $G_2$ will not be active because the regulatory strength of $G_3$ is greater than the regulatory strength of $G_1$, contradicting the expected activation by a hidden inhibition. Hence, pushed to the limit, this consideration prevents the ability to describe any behaviour causally because any programmed action can be unexpectedly preempted by an external one. 

However, by assuming that the design always describes a new functionality which is not observed naturally, the effect becomes the event indicating the effectiveness of a causal relation. As no cause external to the description can trigger the effect, the over-determination by unknown causes is prevented, then insuring that the program is the sole device entailing the expected effect in the biological system.
Hence, the definition of the causal dependence will be governed by the effect leading to the following definition of the dependence: \emph{``if effect $e$ would occur then $c$ occurs''}. Moreover, the scope of future (resp. past) is narrowed to a \emph{closest future (resp. past) period}, representing the fact that a response is always expected in a given delay. Notice that, the proposed definition circumvents the afore mentioned problem illustrated by the hidden inhibition because if the effect does not occur the question of the existence of a cause is meaningless. This definition is somehow equivalent to the causal claims proposed by Lewis~\cite{Lewis2000} in terms of counter-factual conditionals, \ie ``If $c$ had not occurred, $e$ would not have occurred''.
 
Three behavioural dependences are defined in \textsc{gubs}: the \emph{normal} denoted by $\cause$, \emph{persistent} by $\pcause$, and \emph{remanent} by $\remcause$. Informally, for normal dependence the cause precedes the effect providing the effect is observed; for persistent dependence the cause still precedes the effect but it is maintained while the effect is observed; and for remanent dependence, the effect is maintained despite the cause has disappeared. These dependences symbolize common biological interactions. For instance,
in genetic engineering, the recombination enables the emergence of a regulated gene or an hereditary trait permanently. A such mechanism typifies the remanent dependence in biology. The relations between gene expression at steady state are symbolized by persistent dependence. The behavioural dependences are defined as follows (see Section~\ref{sec:semantics} for their formalization): 
\begin{itemize}
\item $c \cause e$: if $e$ occurs then $c$ occurs in the closest past.
\item $c \pcause e$: if $e$ occurs then $c$ occurs in the closest past and also currently.
\item $c \remcause e$: if $e$ occurs then, either $e$ occurs in the closest past or the dependence complies to the property of the normal dependence.
\end{itemize}
Figure~\ref{fig:causes} exemplifies the correspondence between experimental traces, symbolic traces and the history for the causal dependences. 
 All the dependences are extended to a set of causes and a set of consequences, \ie $c_1 + \ldots + c_n \cause e_1+ \ldots + e_m$. 
For example, let us define the activation and the inhibition as follows:
$ g_1 \activate g_2 \equiv g_1 \pcause g_2, \overline g_1 \cause \overline g_2$ and $g_1 \inhibit g_2 \equiv\overline g_1 \pcause g_2, g_1 \cause \overline g_2 $, 
the program depicting a negative regulatory circuit with two genes, \ie $g_1 \activate g_2, g_2 \inhibit g_1$, is:
$ \{ g_1 \pcause g_2, \overline g_1 \cause \overline g_2,\overline g_2 \pcause g_1, g_2 \cause \overline g_1 \}.$

The \emph{observation spots} describe the set of observations expected in a trace. For instance, observing that gene $G$ is at level high is written $\textit{Obs}\textbf{::}G(\textit{High})$. As the activation of a dependence lies on the observation of the effect, the observation spot is used to determine which effects must be necessarily observed. For example, in the negative regulatory circuit, the characteristic observation spots are: $\textit{obs}_1\textbf{::} g_1 + \overline g_2, \textit{obs}_2\textbf{::}\overline g_1+ g_2$.

\paragraph{Compartment \& Context.} A compartment encloses a set of dependences making them local to the compartment. For instance, $ C \{ g_1 \cause g_2 \}$ describes a normal dependence occurring in compartment $C$. The compartments are hierarchically organized and all the compartments are included in another except for the outermost one.
Although the compartments directly refer to the compartmentalized cellular organization (\eg nucleus, mitochondria), they are also used to emphasize the isolation of some interactions by syntactically enclosing the dependences into a compartment. $C.s$ refers to an agent state in compartment $C$. 

A context refers to a stimulus acting on the system, as environmental conditions or external signalling. The application of a context $c$ to a set of dependences $b$ is written $[c]b$ where $c$ is either a variable or a constant. This means that dependences of $b$ are triggered when the context $c$ is present. For instance, recently Ye et al.~\cite{Ye2011} explore the opto-genetics signalling to control the expression of target transgenes. The blue-light induces the expression of transgene ($tg$) via a signalling cascade leading to the binding of \textsc{nfat} transcription factor to a specific promoter (\textsc{pnfat}). The following program using a context summarizes the process: $\context{\text{BlueLight}}\{ \text{NFAT}\pcause tg\}$. A context can be decomposed to several contexts, $[k_1,\ldots,k_n]b$, meaning that all the conditions must be met to trigger the dependences of $b$. The interpretation is equivalent to a context cascading, $[k_1][k_2]\ldots[k_n]b$. Moreover, the observation spots and the attribute definition are context insensitive.

\subsection{Semantics of \textsc{gubs}}
\label{sec:semantics}
The interpretation of \textsc{gubs} is a formula such that the set of all the 
models validating it defines all the possible histories complying to the programmed behaviour.
The interpretation is based on multi-modal hybrid logic with the ``Always'' operator, $\mathcal H (\textbf A, @)$.

\paragraph{Hybrid logic.} In what follows, we recall the formal syntax and semantics of hybrid logic.
 The hybdrid logic~\cite{Blackburn2006,Brauner2010} offers the possibility to denominate worlds by new symbols called \textit{nominals}. They will be used in satisfaction modal operators $@_a$; the formula $@_a\phi$ asserts that $\phi$ is satisfied at the unique point named by the nominal $a$ identifying a particular truth values of a formula at this point. 
Given a set of propositional symbol, $\textsf{PROP}$, a set of relational symbol $\textsf{REL}$, and a set of nominal $\textsf{NOM}$ disjoint to 
 $\textsf{PROP}$, a set of well formed formula in the signature of $\tuple{\textsf{PROP}, \textsf{NOM}, \textsf{REL}}$ is defined as follows:
$$ \phi ::= \top \;|\, p \;|\; a \;|\; \neg \phi \;|\; \phi \land \phi \;|\; @_a \phi \;|\; \E{k} \phi \;|\; \E{k}^- \phi \;|\;\textbf{A} \phi.$$ 
with $p \in \textsf{PROP}, a \in \textsf{NOM}$ and $k \in \textsf{REL}$.
 Moreover, the syntax is extended to other logical operators classically \footnote{ $\bot = \neg \top, \psi \lor \phi = \neg (\neg \psi \land \neg \phi), \psi \to \phi = \neg ( \psi \land \neg \phi), \A{k} \phi = \neg \E{k} \neg \phi, \textbf{E} \phi = \neg \textbf{A} \neg \phi$.}:
 $\bot, \lor, \to, \A{k}, \textbf{E}$.

The interpretation is carried out using the Kripke model satisfaction definition (Table~\ref{tab:hybrid}).
$\mathcal M,w \Vdash \phi$ is interpreted as the satisfaction of a formula $\phi$ by a model $\mathcal M$ at world $w$ where $\Vdash$ stands for the realizability relation (\ie ``is a model of''). A model \emph{validates} a formula, denoted by $ \mathcal M \Vdash \phi$, if and only if it is satisfied for all the worlds of the model (\ie $\forall w \in \dom \mathcal M: \mathcal M,w \Vdash \phi$).
\begin{definition}[Kripke model]
\label{def:kripke}
A Kripke model is a structure $\mathcal M = \tuple{W, (R_k)_{k \in \tau}, V} $ where $W = \dom{\mathcal M}$ is a non-empty set of \emph{worlds}, $\tau \subseteq \textsf{REL} $ a subset of relational symbols denoting the \emph{modalities}, $R_k \subseteq W \times W, k \in \tau$ a relation of accessibility, $V:(\textsf{PROP} \cup \textsf{NOM}) \to \textbf{2}^W$ an interpretation attributing to each nominal and propositional variable a set of worlds such that any nominal addresses one world at most (\ie $\forall a \in \textsf{NOM}: |V(a)| \leq 1$).

\medskip
\noindent
By convention, $R$ stands for the union of the accessibility relation, $R=(\bigcup_{k \in \tau} R_k)$.
\end{definition}
A \emph{modal theory} of a model $\mathcal M$ regarding to a set of formulas $F$, $\text{TH}_F (\mathcal M)$, is the set of formulas of $F$ validated by $\mathcal M$, \ie $\text{TH}_F (\mathcal M)=\{ \phi \in F \mid \mathcal M \Vdash \phi \}$. 
$\textsf{KS}(\phi)$ denotes the set of all models validating $\phi$, \ie $\textsf{KS}(\phi) = \{ \mathcal M \mid \mathcal M \Vdash \phi \}$.

\begin{table}[ht]
\small
$$
\begin{array}{ l @{\; {\mathcal M}, w \Vdash \;} l @{\; \text{iff} \; \; } l}
& \top & true \\
& a & w \in V(a), \; a \in \textsf{NOM} \cup \textsf{PROP} \\
& \neg \phi & {\mathcal M}, w \nVdash \phi\\
& \phi_1 \land \phi_2 & {\mathcal M}, w \Vdash \phi_1 \; \text{and} \; {\mathcal M}, w \Vdash \phi_2 \\
\end{array}
\vrule
\begin{array}{ l @{\; {\mathcal M}, w \Vdash \;} l @{\; \text{iff} \; \; } l}
& @_a \phi & \exists w' \in W: {\mathcal M}, w' \Vdash \phi\;\text{and}\; \{w'\} = V(a) \\
& \E{k}\phi & \exists w' \in W: {\mathcal M}, w' \Vdash \phi \; \text{and}\; w R_k w' \\
& \E{k}^- \phi & \exists w' \in W: {\mathcal M}, w' \Vdash \phi \; \text{and}\; w' R_k w \\
& \textbf {A} \phi & \forall w' \in W: {\mathcal M}, w' \Vdash \phi \\
\end{array}
$$
\label{tab:hybrid}
\caption{Hybrid logic interpretation.}
\end{table}
\paragraph{Semantics.}
A \textsc{gubs} program is interpreted by a hybrid logic formula where the modal operators characterize here the temporal observations on an history:
 $\A{~}$ means \emph{``observed in all the closest futures''} and $\E{~}$ means \emph{``observed in a possible closest future at least''} (resp. $\E{~}^-, \A{~}^-$ for the closest past). Moreover, we assume that the accessibility relations, $(R_k)_{k \in \tau}$, are indexed by the non empty parts of the set of all the contexts of a program $P$, denoted by $K_P$ (\ie $\tau = \partset{K_P} \setminus \{ \emptyset \}$). Then, a non-empty set of contexts ,$\emptyset \subset K \subseteq K_P$, is a modality, \ie $\E{K}, \A{K}$ with $\E{~} = \E{\emptyset}$ by convention.
Let $\tuple{\textsf{W}, \bullet, \Lambda}$ be the set of words $\textsf{W}$ with the concatenation operation and the neutral element, the empty word $\Lambda$ and $\hybridset$ the set of well-formed formulas of ${\mathcal H}(\textbf{A},@)$, the semantics is defined by four functions:
$\sem{. }: \textsf P \to \textsf{F}_{\mathcal H}, \sem{.}_P : \textsf P \to \textsf{W}\to \textbf{2}^{\textsf{W}} \to \textsf{F}_{\mathcal H},
\sem{.}_B : \textsf B \to \textsf{W} \to \textsf{F}_{\mathcal H}, \sem{.}_R : \textsf R \to \textsf{W} \to \textsf{F}_{\mathcal H}$, where $\textsf{P},\textsf{B},\textsf{R}$ respectively stand for the set of \textsc{gubs} programs, the set of agent state set and the set of relations on attributes. $\sem{. }$ initiates the interpretation. Table~\ref{tab:semantics} defines these functions. 
\begin{table}[ht]
\footnotesize
$$ 
\begin{array}{l @{\; = \; } l}
\sem{ \{ b \} } & \textbf{A}\left ( \sem{b}_P(\Lambda)(\emptyset) \right)\\
\multicolumn{2}{c}{ }\\ 
\sem{ \epsilon }_P(C)(K) & \top \\
\sem{ b_1,b_2 }_P(C)(K) & \sem{b_1}_P(C)(K) \land \sem{b_2}_P(C)(K) \\
\sem{s_1 \cause s_2 }_P(C)(K) & \sem{s_2}_B(C) \to \E{K}^-\left (\sem{s_1}_B(C) \right)\\
\sem{s_1 \pcause s_2 }_P(C)(K) & \sem{s_2}_B(C) \to \left(\sem{s_1}_B(C) \land \E{K}^-\left (\sem{s_1}_B(C)\right) \right) \\
\sem{s_1 \remcause s_2}_P(C)(K) & \sem{s_2}_B(C) \to \left((\E{~}^-\sem{s_2}_B(C)) \lor (\E{K}^- \sem{s_1}_B(C)) \right) \\
\sem{ g_1, \cdots, g_n : \;\{r_1,\cdots, r_m \}}_P(C)(K) & \bigwedge_{i=1}^n \bigwedge_{j=1}^m \sem{r_j}_R(C.g_i) \\
\sem{l \textbf{::} s }_P(C)(K) & @_{l} \sem{s}_B(C)\\
\sem{C' \{ b \} }_P(C)(K) & \sem{b}_P(C. C')(K)\\
\sem{\context{K} \{b\}}_P(C)(K') & \sem{b}_P(C)(K \cup K')\\
\multicolumn{2}{c}{ }\\ 
\sem{s_1 + \ldots + s_n}_B(C) & \bigwedge_{i=1}^n \sem{s_i}_B(C)\\
\sem{C'.s}_B(C) & \sem{s}_B(C.C')\\
\sem{g(a)}_B(C)& C.g_a\\
\sem{g(\overline a)}_B(C) & \neg C.g_a\\
\sem{g}_B(C)& C.g\\
\sem{\overline g}_B(C)& \neg C.g\\
\multicolumn{2}{c}{ }\\ 
\sem{a_1 \prec a_2}_R(g) & g_{a_2} \to g_{a_1} \\ 
\sem{a_1 \napprox a_2}_R(g) & g_{a_1} \to \neg g_{a_2} \land g_{a_2} \to \neg g_{a_1}\\
\sem{a}_R(g) & \top
\end{array} 
$$
\label{tab:semantics}
\caption{Semantics of \textsc{gubs}. In the definition, $a$ represents an attribute, $b$ a behaviour, $g$ an agent, $s$ a set of agent states or an agent state, $r$ a relation on attributes, $C$ a compartment, $K$ a set of contexts and $b$ a set of behaviours (\ie contexts, compartments, dependences, attributes, observation spots).
}
\end{table}
For instance, the program of the negative regulatory network,
$\{ 
g_1 \pcause g_2, 
\overline g_1 \cause \overline g_2, 
g_2 \pcause \overline g_1, 
\overline g_2 \cause g_1, 
obs_1:: g_1+\overline g_2, obs_2::\overline g_1+g_2 \}$, is translated into the following formula:
\begin{equation*}
\small
\textbf{A}(
\begin{array}[t]{l}
g_2 \to ((\E{~}^- g_1) \land g_1) \land 
\neg g_2 \to (\E{~}^- \neg g_1) \land 
g_1 \to ((\E{~}^- \neg g_2) \land \neg g_2) \land 
\neg g_1 \to (\E{~}^- g_2) \land 
\\
@_{obs_1} (g_1 \land \neg g_2) \land @_{obs_2}(\neg g_1 \land g_2) 
\end{array}
\end{equation*}

\paragraph{Consistent history.} Now, we formally define the consistency of the history with regards to models. 
An history is assimilated to a path in a model ending by a world labelled with an observation spot label. 
The set of Kripke-models validating the interpretation of a program $P$, $\textsf{KS}(\sem{P})$, not only contains all the consistent histories, but also the possible histories corresponding to behavioural alterations due to external perturbations. Thus, the compilation generates a device such that all the models validating its interpretation integrate all the observations related to the program, including the consistent and the inconsistent ones.

More precisely, the consistency lies on the identification of the largest number of ``relevant'' events characterizing a complete causal chain described in a program. As an history is also a model, a consistent history should validate the interpretation of the complete causal chain.
 The dependence formula set $F_P$ of a program $P$ corresponds to a set of formulas where each formula is the interpretation of a dependence taken separately with the attributes related to the involved agents.
 By definition of the semantics, any model validating the interpretation of a program also validates each formula of this set. The consistency of an history is then based on the validated formulas of this set by this history. \emph{An history $\mathcal{M}_H$ is consistent for $P$ if and only if no other modal theory of histories based on $F_P$ (\ie $\text{TH}_{F_P} (\mathcal M)$ with $\mathcal M$ as an history), ending with the same labelled world includes the modal theory of this history (\ie $\text{TH}_{F_P} (\mathcal M_H) \nsubseteq \text{TH}_{F_P} (\mathcal M)$).}

\section{Compilation}
\label{sec:compil}
At compile phase, a program is transformed to a structure (\eg a \textsc{dna} sequence) while inserted in a vector cell, should behave according to the programmed specification. The structure will result to an assembly of several devices stored in a library of components (\eg parts registry). As the design relates here to a behavioural/functional description, we need to bridge the gap between structural and functional description. This stage is called the \emph{functional synthesis}. The issue is to select a set of components whose assembly preserves the behaviour of the program. To achieve this goal, a \textsc{gubs} program is associated to each component to describe its behaviour. Thereby, the component assembly corresponds to a program assembly preserving the behaviour of the compiled program.
Preserving a behaviour is laid on a property called the \emph{behavioural inclusion} formalizing the fact that the characteristic observational traits of the specified function must be recognized in traces related to the device experiments. In other words, we can exhibit histories consistent with the programmed behaviour from histories consistent with the device behaviour description. 
The behavioural inclusion is defined from the interpretation of the programs, as a logical consequence (Definition~\ref{def:behinc}). 
\begin{definition}[Behavioral inclusion]
\label{def:behinc}
A program $Q$ \emph{behaviourally includes} another program $P$, if and only if the interpretation of the latter is a logical consequence of the interpretation of the former:
$$P \Sqsubset Q \eqdef \forall \mathcal M: \mathcal M \Vdash \sem{Q} \implies \mathcal M \Vdash \sem{P}.$$
\end{definition}
The behavioural inclusion is a pre-order\footnote{A reflexive and transitive relation.} such that the empty program, denoted by $\epsilon$, is a minimum, meaning that a program with no behaviour can be observed in all traces. And a program whose interpretation equals $\bot$, is a maximum. 
 Figure~\ref{fig:subModel} illustrates the behavioural inclusion on a particular model.
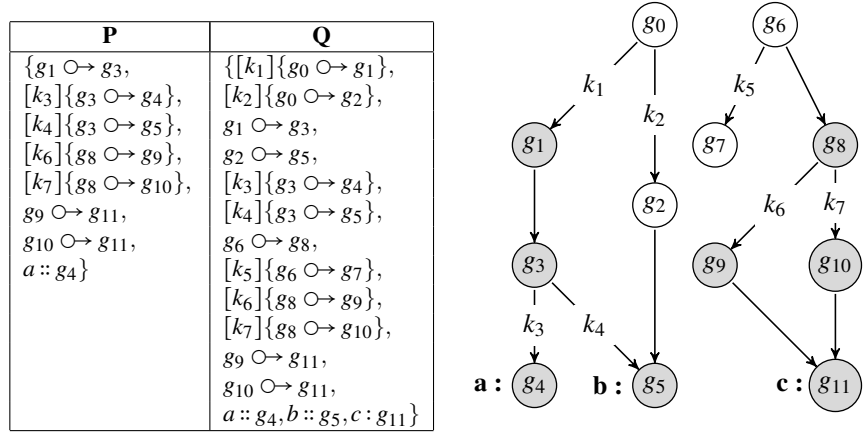
\begin{figure}[t]
	\centering
	\footnotesize 
\begin{tabular}{l r}
	$
	\begin{array}[t]{|l|l|}
	\hline
	\multicolumn{1}{|c|}{\textbf{P}} & \multicolumn{1}{|c|}{\textbf{Q}}\\
	\hline
	\{g_1 \cause g_3,\; & \{\context{k_1}\{g_0 \cause g_1\},\;\\
	\context{k_3} \{g_3 \cause g_4\},\;& \context{k_2}\{g_0 \cause g_2\},\;\\
 \context{k_4} \{g_3 \cause g_5\},\;& g_1 \cause g_3,\;\\
 \context{k_6}\{g_8 \cause g_9\},\;& g_2 \cause g_5,\;\\
 \context{k_7}\{g_8 \cause g_{10}\},\;& \context{k_3}\{g_3 \cause g_4\},\;\\
 g_9 \cause g_{11},\;& \context{k_4}\{g_3 \cause g_5\},\;\\
 g_{10} \cause g_{11},\;& g_{6} \cause g_{8},\; \\
 a:: g_4 \}& \context{k_5} \{g_6 \cause g_7\},\; \\
 & \context{k_6}\{g_8 \cause g_9\},\;\\
 & \context{k_7}\{g_8 \cause g_{10}\},\;\\
 & g_9 \cause g_{11},\;\\
 & g_{10} \cause g_{11},\; \\ 
 & a:: g_4, b:: g_5, c : g_{11} \}\\
\hline
	\end{array}$
	& 
\begin{tikzpicture}[baseline={([yshift={-\ht\strutbox}]current bounding box.north)},scale=0.8]
\small
\GraphInit[vstyle=Normal]
\SetVertexNormal[Shape=circle,LineWidth=0.5pt,LineColor=black]


\Vertex[x=1,y=2,LabelOut=true,Lpos=180,Ldist=1,L={\textbf{a :}},style={minimum size=15}]{P}
\Vertex[x=3,y=2,LabelOut=true,Lpos=180,Ldist=1,L={\textbf{b :}},style={minimum size=15}]{Q}
\Vertex[x=3,y=8,L={$g_0$},style={minimum size=15}]{P0}
\Vertex[x=1,y=6,L={$g_1$},style={fill=black!15,minimum size=15}]{P1}
\Vertex[x=3,y=5,L={$g_2$},style={minimum size=15}]{P2}
\Vertex[x=1,y=4,L={$g_3$},style={fill=black!15,minimum size=15}]{P3}
\Vertex[x=1,y=2,L={$g_4$},style={fill=black!15,minimum size=15}]{P4}
\Vertex[x=3,y=2,L={$g_5$},,style={fill=black!15,minimum size=15}]{P5}


\Vertex[x=5,y=8,L={$g_6$},style={minimum size=15}]{Q0}
\Vertex[x=4,y=6,L={$g_7$},style={minimum size=15}]{Q1}
\Vertex[x=6,y=2,LabelOut=true,Lpos=180,Ldist=1,L={\textbf{c :}},style={minimum size=15}]{R}
\Vertex[x=6,y=6,L={$g_8$},style={fill=black!15,minimum size=15}]{Q2}
\Vertex[x=4,y=4,L={$g_9$},style={fill=black!15,minimum size=15}]{Q3}
\Vertex[x=6,y=4,L={$g_{10}$},style={fill=black!15,minimum size=15}]{Q4}
\Vertex[x=6,y=2,L={$g_{11}$},style={fill=black!15,minimum size=15}]{Q5}
\tikzstyle{EdgeStyle}=[post]
\Edge[label={$k_1$}](P0)(P1)
\Edge[label={$k_2$}](P0)(P2)
\Edge[label={}](P1)(P3)
\Edge[label={}](P2)(P5)
\Edge[label={$k_3$}](P3)(P4)
\Edge[label={$k_4$}](P3)(P5)
\Edge[label={$k_5$}](Q0)(Q1)
\Edge[label={}](Q0)(Q2)
\Edge[label={$k_6$}](Q2)(Q3)
\Edge[label={$k_7$}](Q2)(Q4)
\Edge[label={}](Q3)(Q5)
\Edge[label={}](Q4)(Q5)
\end{tikzpicture}
\end{tabular}
\caption{Behavioral inclusion example. Consistent histories of $P$ necessary contains worlds coloured in gray. }
	\label{fig:subModel}
\end{figure}
\paragraph{Observability.}
\newcommand{\obs}[0]{\textbf{obs}\,}
It may arise that no history will be consistent with a programmed behaviour. For example, the program $\{ \textit{Obs}::g, \overline g \pcause g \}$ is not observable in a trace. Indeed, its interpretation yields to the following formula: $\textbf{A}( (@_{Obs} g) \land (g \to ((\E{~}^- \neg g) \land \neg g)))$, false in all models because world $\textit{Obs}$ must both satisfies $g$ and $\neg g$ by definition of the persistent dependence. A \textsc{gubs} program is said \emph{observable} if and only if the formula resulting from its interpretation is validated by one model at least. Hence, the interpretation of an unobservable program is an antilogy. An unobservable program can be assimilated to a programming error. The detection of such errors can be carried out at compile-phase by using tableaux method~\cite{Cerrito2011}
that automatically determines whether a formula is satisfiable in a model. Indeed \textsc{gubs} uses fragment of \textit{HL(@)} logic which is decidable. Notice that an observable program always behaviourally includes an observable program (Proposition~\ref{prop:comp1}).
\begin{proposition} A program behaviourally included in an observable program is observable: 
\label{prop:comp1}
$\forall P,Q \in \textsf{P}: \obs Q \land P \Sqsubset Q \implies \obs P.$
\end{proposition}
\subsection{Functional synthesis}
The functional synthesis is the operation whereby biological components of a library are selected and assembled to generate a device behaviourally including the designed function. The behaviour of each component is described by a \textsc{gubs} program. At its simplest, the functional synthesis could be considered as a proper substitution of variables by constants. For example, in the following activation $\{G_1 \activate g_2\}$, $g_2$ will be substituted by gene $G_2$, providing that component $Q$ describes the activation $\{ G_1 \activate G_2 \}$. However, more complex situations may arise during component selection. For example, if the activation $G_1 \activate G_2$ occurs with another regulation only \ie $Q=\{G_1 \activate G_2, G_3 \activate G_4 \}$ then the selection of $Q$ adds a supplementary regulation.

Formally, a finite substitution is a set of mappings, $\sigma=\{ v_i / b_i \}_i$, on variables and constants such that a variable can be substituted by a variable or a constant, and a constant can only substituted by itself\footnote{ $P\sigma$ or $P[\sigma]$ represents its application on program $P$ and identity substitutions are omitted.}. For instance, we have:
$\{\textit{Obs}\textbf{::} G(l)+b_2, b_1 \cause G(l) \}[\{b_1 \mapsto B_1,b_2 \mapsto B_2, l \mapsto \textit{Low}\}]= 
\{\textit{Obs}\textbf{::} G(\textit{Low})+B_2, B_1 \cause G(\textit{Low}) \}$.

\paragraph{Functional synthesis rules.} The functional synthesis is defined by rules (Table~\ref{tab:rules}) governing the component assembly. Only the dependences and the attributes will be functionally synthesize. 
The observation spots are considered as annotations used for the compilation process. 
To insure the correctness, each transform must preserved the seminal behaviour. Hence, 
each program resulting from the application of a rule must behaviourally includes the previous one. 
Formally, the functional synthesis is modelled by a relation on programs denoted by $\leftfree$, \ie $Q \leftfree_\sigma P$ where $P$ is the initial program and $Q$ the transformed one, such that each rule insures that:
\emph{$Q \leftfree_\sigma P$ is correct with regards to a substitution $\sigma$, that is 
$P[\sigma] \Sqsubset Q[\sigma]$ and $Q[\sigma]$ is observable.}
Also notice that the behavioural inclusion is preserved by substitution (Proposition~\ref{prop:subst-behinc}). 

 \begin{proposition} For all substitutions $\sigma$, we have: $P \Sqsubset Q \implies P[\sigma] \Sqsubset Q[\sigma]$. 
\label{prop:subst-behinc}
\end{proposition}

\noindent
Table~\ref{tab:rules} describes the functional synthesis rules\footnote{Rules are of the form: \AxiomC{hypothesis}
\UnaryInfC{conclusion}
\DisplayProof. }. 
 $\Gamma$ is a set of components representing the library. 
$P \subseteq_{\text{Asm}} Q$ denotes the fact that program $Q$ corresponds to an assembly including $P$ \ie $Q= (Q_1,P,Q_2)$ where $Q_1$ or $Q_2$ may be an empty program.
\begin{table}[ht]
\begin{center}
\footnotesize
\renewcommand{\arraystretch}{2}
\setlength{\doublerulesep}{\arrayrulewidth}
\begin{tabular}{l r} 
\multicolumn{2}{c}{ \textsc{- Instantiation -}} \\
\multicolumn{2}{c}{
\AxiomC{$Q[\sigma] \subseteq_{\text{Asm}} P[\sigma]$}
\AxiomC{$\obs( Q[\sigma])$}
\AxiomC{$Q \in \Gamma$}
\RightLabel{(Inst.)}
\TrinaryInfC{$ Q \leftfree_\sigma P$}
\DisplayProof}
\\ 
\multicolumn{2}{c}{ \textsc{- Commutativity, Contraction -}} \\
\AxiomC{$Q \leftfree_\sigma P,P'$}
\RightLabel{(Com.)}
\UnaryInfC{$ Q \leftfree_\sigma P',P$}
\DisplayProof 
&
 \AxiomC{$Q \leftfree_\sigma P$}
\RightLabel{(Cont.)}
\UnaryInfC{$ Q \leftfree_\sigma P,P$}
\DisplayProof 
\\ 

\multicolumn{2}{c}{ \textsc{- Assembly -}} \\
\multicolumn{2}{c}{ 
\AxiomC{$Q \leftfree_\sigma P$}
\AxiomC{$Q'\leftfree_{\sigma'} P'$}
\AxiomC{$\restrict{\sigma}{\text{VA}(P) \cap \text{VA}(P')} = \restrict{\sigma'}{\text{VA}(P) \cap \text{VA}(P')} \qquad \obs (Q[\sigma],Q'[\sigma'])$} 
\RightLabel{(Asm.)}
\TrinaryInfC{$Q,Q' \leftfree_{\sigma \cup \sigma'} P,P' $}
\DisplayProof 
} 
\end{tabular}
\end{center}
\caption{ Functional synthesis rules}
\label{tab:rules}
\end{table}
Rule~(Inst.) describes the fact that an observable instance of a part of a component in the library is functionally synthesized. Rule~(Com.) expresses the commutativity of the assembly. Rule (Cont.) contracts the redundant formulation of programs. Finally, Rule~(Asm.) details the conditions for an assembly of two components, each representing a functional synthesis of a part of the designed function. 
A detailed example of their use on a real case is given in Section~\ref{sec:example}.
\begin{theorem} The functional synthesis rules (Table~\ref{tab:rules}) are correct.
\label{the:fsr}
\end{theorem}
\begin{table}[ht]
\begin{center}
\footnotesize
\textsc{- dependences - }
\renewcommand{\arraystretch}{2}
\setlength{\doublerulesep}{\arrayrulewidth}
\begin{tabular}{c c c} 
\AxiomC{$Q \leftfree_\sigma S_1 \pcause S_2,S_2 \pcause S_3, \Delta $}
\RightLabel{(Trans.)}
\UnaryInfC{$Q \leftfree_\sigma S_1 \pcause S_3, \Delta$}
\DisplayProof 
&
\AxiomC{$ Q \leftfree_\sigma S_1 \pcause S_2, \Delta$} 
\RightLabel{(N2P.)}
\UnaryInfC{$Q \leftfree_\sigma S_1 \cause S_2, \Delta$}
\DisplayProof 
&
\AxiomC{$ Q \leftfree_\sigma S_1 \cause S_2, \Delta $} 
\RightLabel{(R2N.)}
\UnaryInfC{$Q \leftfree_\sigma S_1 \remcause S_2, \Delta$}
\DisplayProof 
\end{tabular}

\bigskip
\textsc{- agent states - }

\begin{tabular}{l c r} 
\AxiomC{$ S_1+S_2 $} 
\RightLabel{(SCom.)}
\UnaryInfC{$S_2+S_1 $}
\DisplayProof 
& 
\AxiomC{$S + s$}
\RightLabel{(SCont.)}
\UnaryInfC{$S+s+s$}
\DisplayProof
&
\AxiomC{$S + s$}
\RightLabel{(Incl.)}
\UnaryInfC{$S$}
\DisplayProof
\end{tabular}
\end{center}
\caption{Rules for the dependences and the agent states. $S_i$ stands for a collection, $s_1 + \ldots + s_n$, of agent states, including negation, and $\Delta$ stands for the rest of the program. }
\label{tab:causes}
\end{table}
Another set of rules, more specifically devoted to dependences (Table~\ref{tab:causes}), defines the alternate possibilities to express similar behaviours. The table also includes the rules for agent sets. 
Rule (Trans.) expands the chain of the persistent dependences by adding intermediary dependence to refine a pathway. Rule (N2P.) transforms a normal dependence to a persistent one since the latter is a normal dependence with an additional property. And Rule (R2N.) transforms a remanent dependence to a normal dependence, since normal dependence is also remanent dependence with a repetition of the effect restricted to one step.
According to these rules, all the dependence chains can be implemented with persistent dependences.

A possible algorithm for the assembly could be based on a combinatorial application of the rules. However, such algorithm may reveal inefficient in practice. The conditions for an efficient algorithm of compilation should be based on an internal representation of a program, as a set of contextualized dependences with attributes, $\{ \{A,[K] S_1 \allcause S_2\} \}$, such that $A,K,S_1,S_2$ are respectively: a set of attributes specification related to the agent involved in the dependency, a set of contexts and sets of agent states. Any program can be encoded under this representation from a normal form of the program (not detailed here). Accordingly, the problem solved by the compilation algorithm can be defined as follows (Definition~\ref{def:probasm}):
\begin{definition}[Functional Synthesis Problem]
\label{def:probasm}
Let $\Gamma=\{Q_i\}_{1 \leq i \leq n}$ be set where each $Q_i$ is a set of contextualized dependences with attributes and $P$ a set of contextualized dependences with attribute, can we find the smallest observable subset of components $C\subseteq \Gamma,$ such that there exists a substitution $\sigma$ so that its application on the components of $C$ form a cover of $P[\sigma]$,\ie
$\exists \sigma: P[\sigma] \subseteq \bigcup_{Q_j \in C} Q_j[\sigma] \land \obs C.$
\end{definition}
As the set cover problem is reducible to this problem, the problem is NP-complete. Then, the resolution is oriented towards a heuristic algorithm. 

\section{Example}
\label{sec:example}
\begin{figure}[ht]
\newcommand{\labelstyle}[1]{\textit{\scriptsize #1}}
\centering
	\begin{tikzpicture}[scale=0.7]
	\small
\GraphInit[vstyle=Normal]
\SetVertexNormal[Shape=circle,LineWidth=0.5pt,LineColor=black]
\SetUpEdge[lw =1.5pt, labelcolor = white, labeltext = black, labelstyle = {sloped, text= black,above}]
\Vertex[x=-2.5 ,y=6,L={\labelstyle{\scriptsize Tetr}},style={minimum size=30pt}]{P0}
\Vertex[x=-0 ,y=6,L={\labelstyle{LuxR}},style={minimum size=30pt}]{P1}
\Vertex[x=2.5 ,y=5,NoLabel=true, style={minimum size=3pt, inner sep=0pt, outer sep=4pt, shape=circle,blue}]{A3}
\Vertex[x=2.2 ,y=6.,NoLabel=true, style={minimum size=3pt, inner sep=0pt, shape=circle,blue}]{P3}
\Vertex[x=2 ,y=5.5,NoLabel=true, style={minimum size=3pt, inner sep=0pt, outer sep=4pt, shape=circle,blue}]{A4}
\Vertex[x=2.7,y=6.2,NoLabel=true, style={minimum size=3pt, inner sep=0pt, shape=circle,blue}]{P5}
\Vertex[x=3.2 ,y=5.3,NoLabel=true,style={minimum size=3pt, inner sep=0pt, shape=circle,blue}]{A0}
\Vertex[x=2.8 ,y=5.7,NoLabel=true,style={minimum size=3pt, inner sep=0pt, outer sep=5pt,shape=circle,blue}]{A2}

\Vertex[x=2.5 ,y=5,LabelOut=true,Lpos=270,Ldist=1,L={$AHL$},style={minimum size=3pt, inner sep=0pt, shape=circle,blue}]{P8}
\Vertex[x=5 ,y=6,L={\labelstyle{LuxR}},style={minimum size=30pt}]{Q0}
\Vertex[x=7.5 ,y=6,L=\labelstyle{$Lacl_{M1}$},style={minimum size=30pt}]{Q1}
\Vertex[x=6 ,y=4,L={\labelstyle{Cl}},style={minimum size=30pt}]{Q2}
\Vertex[x=8.5 ,y=4,L={\labelstyle{Lacl}},style={minimum size=30pt}]{Q3}
\Vertex[x=10 ,y=6,L={\labelstyle{GFP}},style={minimum size=30pt}]{Q4}
\tikzstyle{EdgeStyle}=[color=black!70,post]
\Edge[label={$+$}](P0)(P1)
\Edge[label={$+$}](P1)(A4)
\Edge[label={$+$}](A2)(Q0)
\Edge[label={$+$}](Q0)(Q1)
\Edge[label={$+$}](Q0)(Q2)
\Edge[label={$-$}](Q2)(Q3)
\Edge[label={$-$}](Q1)(Q4)
\Edge[label={$-$}](Q3)(Q4)
\end{tikzpicture}
\caption{The band detector regulatory circuit.}
	\label{fig:bandDetect}
\end{figure}
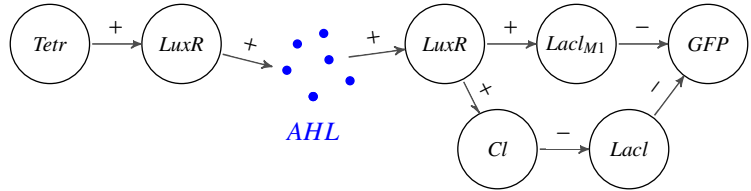

The compilation process is here exemplified in a real case by the design of the Band Detector proposed in \cite{Basu2005}.
This example explains how from a simple abstract definition of the functionality a complex design can be synthesized.
Accordingly, \textsc{gubs} may be used to describe a behaviour with a high-level of programming well as a low-level, detailing the components
involved in the process. Although, the functional synthesis is not yet performed automatically, it is worth to point out that the different transforms of the high-level program to obtain the final design complies to rules of Tables~\ref{tab:rules},~\ref{tab:causes}, insuring its correctness and so, its functional safety in the context of open system. 
 
 The design aims at forming patterns of different colours in a population of bacteria exploiting the quorum sensing phenomenon by staining with fluorescent protein (GFP).
The amount of molecules of interest that receives a cell depends on its relative position to the cell diffusing the molecule of interest controlled by an external event: more the cell is far from the source, the fewer is the amount of molecules received. 
The activation or inhibition of the fluorescent protein due to the concentration will distinguish the bands surrounding the source. In the original design, the protein does not fluoresce in an intermediary band. 

From a computing standpoint, we can assimilate the design to a message transmission coupled to a sensor/actuator responsible for fluorescence, then leading to a concise \textsc{gubs} program presented below: the diffusive molecule is \textit{AHL} which production is controlled by a context and the observation is applied on \textit{GFP}. Two categories of cells are defined: the \emph{Sender} and the \emph{Receiver}. Therefore, two \textsc{gubs} programs identify the two cell types.
 \begin{equation*}
\scriptsize
\renewcommand{\arraystretch}{1.5}
\setlength{\doublerulesep}{\arrayrulewidth}
\begin{array}{l @{=} l} 
\textit{Sender} & \{\; \; \textsf{AHL:}\{low \napprox mid \napprox high \},
 \context{Light}\{detect\cause\textsf{AHL}(low),detect\cause\textsf{AHL}(mid),detect\cause\textsf{AHL}(high)\} \}\\
\textit{Receiver}&\{\; \; \textsf{AHL}(low)\cause \overline{\textsf{GFP}}, \textsf{AHL}(mid)\cause \textsf{GFP}, \textsf{AHL}(high)\cause \overline{\textsf{GFP}}, obs_1\textbf{::}\textsf{GFP}, obs_2 \textbf{::}\overline{\textsf{GFP}} \}
\end{array} 
\end{equation*}

Figure ~\ref{fig:bandDetect} describes the original genetic circuit used in the article. The diffusible molecule is the constant \emph{AHL}.
The gene \textit{LuxR} has three activation thresholds: at Level $2$, it activates both \textit{LaclM1} and \textit{Cl}, at level $1$, the amount of \textit{AHL} only allows activation of \textit{Cl}, and finally, at level $0$, none are activated.
\begin{table}[ht]
\footnotesize
\centering
$$
\begin{array}{l @{=} l} 
Q_1& \{\context{\text{Light}}\{detect \cause\textsf{Tetr}\}\} \\
Q_2& \{\textsf{Tetr}\activate\textsf{Luxl}\}\\
Q_3& \{\textsf{AHL:}\{low \napprox mid \napprox high \} , \textsf{Luxl}\activate\textsf{AHL}(low),\textsf{Luxl}\activate\textsf{AHL}(mid),\textsf{Luxl}\activate\textsf{AHL}(high) \}\\
Q_4& \{\textsf{AHL:}\{low \napprox mid \napprox high \} , \textsf{LuxR:}\{low \napprox \{mid \prec high\}\}, \textsf{AHL}(mid)\cause\textsf{LuxR}(mid),\textsf{AHL}(high)\cause\textsf{LuxR}(high)\}\\
Q_5& \{\textsf{LuxR:}\{low \napprox \{mid \prec high\} \}, \textsf{LuxR}(mid)\activate\textsf{Cl},\textsf{LuxR}(high)\activate\textsf{Cl}+\textsf{LaclM1} \}\\
Q_6& \{\textsf{Cl}\inhibit\textsf{Lacl}\}\\
Q_7 & \{\textsf{LaclM1}\inhibit\textsf{GFP}\}\\
Q_8 & \{\textsf{Lacl}\inhibit\textsf{GFP}\}
\end{array}
$$
\caption{Part of the database dedicated to the Band Detector.}
\label{tab:bddWeiss}
\end{table}
We show that from the sender-receiver program, we obtain the original design by applying the afore mentioned rules
 with an appropriate selection of components. The regulations of Figure~\ref{fig:bandDetect} are described in \textsc{gubs} program
(Table~\ref{tab:bddWeiss}) translating in term of dependences and relations on their attributes their regulatory action. We focus here on some illustrative steps of the sender program compilation. The complete functional synthesis is given in Appendix. The compilation consists in finding the appropriate components whose assembly behaviourally includes the sender-receiver program, with the particularity that the diffusive molecule must be the same in both programs. To ease compilation follow-up, we label each dependency of the sender-receiver program (Table~\ref{tab:reecCode}).
\begin{table}[t]
\scriptsize
\centering
$$
\begin{array}{|l @{\; = \; } l | l @{\; = \; } l|} 
\hline
\multicolumn{2}{|c|}{\mathit{Sender}} & \multicolumn{2}{c|}{\mathit{Receiver}} \\
\hline 
P_{11} & \{\context{Light}\{detect\cause\textsf{AHL}(low)\}\} & P_{21} & \{\textsf{AHL}(low)\cause \overline{\textsf{GFP}}\} \\
P_{12} & \{\context{Light}\{detect\cause\textsf{AHL}(mid)\}\} & P_{22} & \{\textsf{AHL}(mid)\cause \textsf{GFP}\}\\
P_{13} & \{\context{Light}\{detect\cause\textsf{AHL}(high)\}\}& P_{23} & \{\textsf{AHL}(high)\cause \overline{\textsf{GFP}}\}\\
\hline
\end{array}
$$
with $\{\textsf{AHL:}\{low \napprox mid \napprox high \}\}$ as attributes of \textit{AHL}.
\caption{Separation of the dependences.}
\label{tab:reecCode}
\end{table}
Let us consider $P_{11}$ whose compilation is closed to $P_{12}$ and $P_{13}$. Notice that $P_{11}$ cannot be directly instantiated with any component because, in the one hand, the component $Q_1$ contains a context like $P_{11}$ but applied on gene $\mathit{Tetr}$ instead of \textit{AHL}, and on the other hand $Q_3$ has the \textit{AHL} molecule but no context is defined. So, to fit $P_{11}$ with the components $Q_1$, $Q_2$ and $Q_3$, first, the normal dependence is converted to persistent one (Rule~(N2P.)).
\begin{equation*}
\scriptsize
\AxiomC{$Q_1,Q_2,Q_3\leftfree_{\sigma} \{\context{light}\{detect\pcause AHL(low)\} \}$}
\RightLabel{(N2P.)}
\UnaryInfC{$Q_1,Q_2,Q_3\leftfree_{\sigma} P_{11}$}
\DisplayProof
\end{equation*}
Thereby, the resulting dependence can be separated to match the assembly $Q_1, Q_2, Q_3$ by applying (Trans.) rule twice. $v_1$ and $v_2$ are fresh variables.
\begin{equation*}
\scriptsize
\AxiomC{$Q_1,Q_2,Q_3\leftfree_{\sigma} P_{11}'= \{\context{light}\{detect\pcause v_2, v_2\pcause v_1, v_1\pcause AHL(low)\}$}
\RightLabel{(Trans.)}
\UnaryInfC{$Q_1,Q_2,Q_3\leftfree_{\sigma} \context{light}\{detect\pcause v_1, v_1\pcause AHL(low)\}$}
\RightLabel{(Trans.)}
\UnaryInfC{$Q_1,Q_2,Q_3\leftfree_{\sigma} \context{light}\{detect\pcause AHL(low)\}$}
\DisplayProof
\end{equation*}
Finally, we obtain a new program program $P_{11}'$ compatible with $Q_1,Q_2,Q_3$, and
each variable is substituted by a constant (biological element) with the application of Rule~(Inst.). For $P'_{11}$ we have:
\begin{equation*}
\scriptsize
\AxiomC{$Q_1,Q_2,Q_3[\sigma=\{light/Light,v_1/Tetr, v_2/Luxl\}] \subseteq_{Asm} P_{11}'[\sigma]$}
\AxiomC{$obs(Q_1,Q_2,Q_3[\sigma])$}
\RightLabel{(Inst.)}
\BinaryInfC{$Q_1,Q_2,Q_3\leftfree_{\sigma} \context{light}\{detect\pcause v_1, v_1\pcause v_2, v_2 \pcause AHL(low)\}$}
\DisplayProof
\end{equation*}
By following this scheme for $P_{12}$ and $P_{13}$, we respectively obtain $P_{12}'$ and $P_{13}'$.
The final assembly corresponds to the functional synthesis of \textit{Sender} program. 
\begin{equation*}
\scriptsize
\AxiomC{$Q_1,Q_2,Q_3\leftfree_{\sigma} P_{11}'$}
\noLine
\UnaryInfC{$\vdots$}
\noLine
\UnaryInfC{$Q_1,Q_2,Q_3\leftfree_{\sigma} P_{11}$}
\AxiomC{$Q_1,Q_2,Q_3\leftfree_{\sigma} P_{12}'$}
\noLine
\UnaryInfC{$\vdots$}
\noLine
\UnaryInfC{$Q_1,Q_2,Q_3\leftfree_{\sigma'} P_{12}$}
\AxiomC{$Q_1,Q_2,Q_3\leftfree_{\sigma} P_{13}'$}
\noLine
\UnaryInfC{$\vdots$}
\noLine
\UnaryInfC{$Q_1,Q_2,Q_3\leftfree_{\sigma''} P_{13}$}
\RightLabel{(Asm.)}
\TrinaryInfC{$Q_1,Q_2,Q_3\leftfree_{\sigma\cup{\sigma}'\cup{\sigma}''} P_{11},P_{12},P_{13}$}
\DisplayProof
\end{equation*}
In conclusion, the functional synthesis generates the original genetic circuit (Figure~\ref{fig:bandDetect}) from the sender program. A similar approach can be also applied to obtain the receiver program (see the complete proof in Appendix~\ref{sec:proofEx}).
\begin{equation*}
\scriptsize
\begin{array}{l c l} 
\textit{Sender}&= & \{\textsf{AHL:}\{low \napprox mid \napprox high \},\context{\text{Light}}\{detect\cause\textsf{Tetr}\},\\
&& \; \;\textsf{Tetr}\activate\textsf{Luxl}, \textsf{Luxl}\activate\textsf{AHL}(low),\textsf{Luxl}\activate\textsf{AHL}(mid),\textsf{Luxl}\activate\textsf{AHL}(high)\}\\
\end{array} 
\end{equation*}

\section{Related works}
\label{sec:related}
Several domain specific languages have been developped to model and simulate biological systems. 
Based on process-calculus, seminally used to model process concurrency, several rule-based languages model protein interactions~\cite{Priami2001,Danos07,Ciocchetta2009}. Another approach is based on logic, such as \textsc{biocham}~\cite{Calzone2006} that formalizes the temporal properties of a biological system. As these languages are dedicated to simulation, the objective is to close the systems because the simulations need to integrate all the characteristics of the analysed systems. By comparison, the purpose of \textsc{gubs} is different since the issue is to represent the behaviour of a synthetic device in an organism, leading to translate the notion of the openness of biological systems by the semantics of the language.

In synthetic biology, the structural description languages ~\cite{Czar2009,Pedersen2009a,Bilitchenko2011} allow to specify 
 well-formed genome sequences by grammars modularly and hierarchically. Although the sequence description is necessary, the programmer must previously anticipate the behaviour of the device to conceive. Besides, the behavioural design is not included in the program while it initially motivates it. In \textsc{gubs}, the design is driven by a behaviour description and sequence selection is postponed at compile phase. Moreover, the size of the structural description is also subject to a combinatorial explosion when the complexity of programmed systems increases.

Amorphous programming language has been also investigated to specify the biological devices at the scale of cell colony, here considered as a possible computing medium for amorphous program. J.~Beal~\cite{Beal2011} demonstrates the proof of concept of this approach in \textsc{proto}, showing the feasibility of an automatic compile chain. In \textsc{gubs}, the compile chain is based on rewriting rules whose correctness have been formally proved with regards to a semantics describing the constraints of an open system.

Developing a language for biological systems actually involves to consider several unknown due to their openness: lack of knowledge on all the interactions in biological circuits and imprecise definition of initial conditions. We only know the result of a chain of effects. Then, the major constraint for programming open system seems to be: how to provide an expressive language to describe the dynamics of such systems, but simple enough to capture the essence of the biological questions in a small program in order to allow programming of large biological systems with a program humanly achievable.

In the future, the design in synthetic biology will certainly require different programming layouts based on different paradigms addressing the integration levels of biological systems. In a tower of languages, starting from a language with collective operations on cell colony, using an amorphous programming language as Proto~\cite{Beal2011} or a language for dynamical systems with dynamical structures as MGS~\cite{Spicher2005}, and ending by a structural description programmed in a grammar based language, \textsc{gubs} language occupies the intermediary level dedicated to cell entity behavioural programming.

\section{Conclusion}
\label{sec:conclusion}
In \textsc{gubs} language, we propose to characterize a programming paradigm abstracting the molecular interactions in the context of open system, that differs to an approach dedicated to biological system modeling. Accordingly, the interactions are symbolized by causal dependences whose interpretation is driven by effect. We have demonstrated the proof-of-concept of the compilation based on rewriting rules, and illustrated it on a realistic example. The perspective of this work is to find an efficient compilation algorithm. 
Identifying the biological parameters guiding the component selection should be a key issue in this undertaking.

\paragraph{Acknowledgements.}
The funding for most of this work is granted by the \textsc{anr synbiotic} (\textsc{anr blan} 0307 01) and we would like to thank the colleagues of this project for their fruitful discussions. 
\bibliographystyle{eptcs}
\bibliography{gub}	

\vfill 

\section*{Appendix}
\begin{table}[ht]
$$
\small
\begin{array}{l @{\; ::= \;} l}
\text{program} & \{ \text{behaviour} \} \\
\text{behaviour} & \text{behaviour}, \text{behaviour} \;|\; \text{behaviour} \\ 
\text{behaviour} & \text{compartment} \; |\; \text{dependence} \;|\; \text{context} \;|\; \text{observation} \;|\; \text{defattributes}\\ 
\text{compartment} & \text{varconstant} \; \{\text{behaviour} \}\\
\text{observation} & \text{varconstant} \textbf{::} \text{worlds} \\
\text{context} & \context{\text{varconstants}} \; \{\text{behaviour} \}\\
\text{dependence} & \text {worlds} \cause \text{worlds} \; | \; \text{worlds} \pcause \text{worlds} \; | \; \text{worlds} \remcause \text{worlds} \\
\text{world } & \text{attribute} \;|\; \text{varconstant} (\text{attribute}) \;|\; \text{varconstant}.\text{world } \\
\text{worlds} & \text{worlds} + \text{world } \;|\; \text{world } \\
\text{attribute} & \text{varconstant} \;|\; \overline{\text{varconstant}}\\
\text{defattribute} & \text{varconstants} : \text{attspec} \\ 
\text{attspec} & \text{defspec} \{ \text{varconstants} \} \;|\; \{ \text{attrels} \} \\
\text{defspec} & \textsf{exclusion} \;|\; \textsf{inclusion} \\
\text{attrels} & \text{attrels}, \text{attrel} \;|\; \text{attrel} \\
\text{attrel} & \text{varconstant} \prec \text{varconstant} \;|\; \text{varconstant} \napprox \text{varconstant} \;|\; \text{varconstant}\\
\text{varconstant} & \textit{word} \; | \; \textit{Word}\\
\text{varconstants} & \text{varconstants}, \text{varconstant} \;|\; \text{varconstant} 
\end{array}
$$
\label{tab:syntax}
\caption{Syntax of \textsc{gubs} program}
\end{table}
\subsection*{Proofs}
\begin{proof}[Proposition~\ref{prop:comp1}]
By contradiction, assume that $P$ is unobservable, then there does not exist a model satisfying the formula. As $Q$ is observable, we deduce that there exists models satisfying $Q$, but no restricted model must satisfy $P$, that contradicts the definition of the behavioural consequence.
\end{proof}

\begin{proposition}
\label{prop:correct1}
Let $\psi \in \textsf{F}_{\mathcal H}$ be a formula, let $\sigma: (\textsf{NOM} \cup \textsf{PROP} \cup \textsf{REL} ) \to (\textsf{NOM} \cup \textsf{PROP} \cup \textsf{REL})$ be a substitution on nominals, variables and relational symbols, let $\mathcal M=\tuple{W, (R_k)_{k \in \tau}, V}$ be a model, we define the model $\mathcal{\tilde M}=\tuple{W, (\tilde R_{k})_{k \in \tilde \tau}, \tilde V}$ from $\mathcal M$ as follows:

\begin{enumerate}
\item $\forall a \in \textsf{NOM} \cup \textsf{PROP}, \forall w \in W: w \in V(a\sigma) \iff w \in \tilde V(a) $ \label{p1}
\item $\forall k \in \tilde \tau : w R_{k \sigma} w' \iff w \tilde R_{k} w';$ \label{p2}
\end{enumerate}
 we have: $ \mathcal M,w \Vdash \psi \sigma \iff \mathcal{\tilde M},w \Vdash \psi.$

\begin{proof}
The proof is defined by induction on the formula:

without loss of generality, we assume that $\psi$ is in Negation Normal Form where negation occurs only immediately before variables only. Recall that every formula can be set in Negation Normal Form. 
\begin{itemize}
\item $\mathcal{ M},w \Vdash a \iff \mathcal{\tilde M},w \Vdash a, a \in \textsf{PROP} \cup \textsf{NOM} $. By (\ref{p1}), we have $w \in V(a\sigma) \iff w \in \tilde V(a)$ leading to the equivalence.

\item $\mathcal{ M},w \Vdash \neg a \iff \mathcal{\tilde M},w \Vdash \neg a$. By definition of the realizability relation, this is equivalent to: $\mathcal{\tilde M},w \nVdash a \iff \mathcal{\tilde M},w \nVdash a$. By (\ref{p1}), this equivalence holds.

\item $\mathcal{ M},w \Vdash (\psi_1 \land \psi_2)\sigma \iff \mathcal{\tilde M},w \Vdash (\psi_1 \land \psi_2)$. 
 By definition of the substitution, we have to prove:
 $\mathcal{ M},w \Vdash (\psi_1\sigma) \land (\psi_2\sigma) \iff \mathcal{\tilde M},w \Vdash (\psi_1 \land \psi_2).$
By definition of the realizability relation we can formulate the property equivalently as follows: 
 $$\mathcal{ M},w \Vdash (\psi_1\sigma) \land \mathcal{\tilde M},w \Vdash (\psi_2\sigma) \iff \mathcal{\tilde M},w \Vdash \psi_1 \land \mathcal{\tilde M},w \Vdash \psi_2.$$
By induction hypothesis, we have: $\mathcal{\tilde M},w \Vdash (\psi_1\sigma) \iff \mathcal{\tilde M},w \Vdash \psi_1$ and $\mathcal{\tilde M},w \Vdash (\psi_2\sigma) \iff \mathcal{\tilde M},w \Vdash \psi_2$, implying the previous condition.

\item $\mathcal{M},w \Vdash (\psi_1 \lor \psi_2)\sigma \iff \mathcal{\tilde M},w \Vdash (\psi_1 \lor \psi_2)$. The proof is similar to the proof of the previous item ($\land$).

\item $\mathcal{M},w \Vdash (@_a\psi )\sigma \iff \mathcal{\tilde M},w \Vdash @_a \psi.$
By definition of the substitution we have to prove that: 
$\mathcal{M},w \Vdash (@_{a\sigma} \psi\sigma ) \iff \mathcal{\tilde M},w \Vdash @_a \psi $
By definition of the realizability relation, this is equivalent to:
$$\exists w' \in W: w \in V(a\sigma) \land \mathcal M,w' \Vdash \psi \sigma \iff \exists w'' \in W: w'' \in \tilde V(a)\sigma \land \mathcal{\tilde M},w'' \Vdash \psi.$$
By setting $w'=w''$, from (\ref{p1}) we have: $w' \in V(a \sigma) \iff w' \in V(a)$.
By induction hypothesis, we have: $\mathcal M,w' \Vdash \psi \sigma \iff \mathcal{\tilde M},w' \Vdash \psi.$
The both last properties imply that:
$$\exists w' \in W: w \in V(a\sigma) \land \mathcal M,w' \Vdash \psi \sigma \iff \exists w' \in W: w' \in \tilde V(a)\sigma \land \mathcal{\tilde M},w'' \Vdash \psi,$$
which implies the initial property.

\item $\mathcal{M},w \Vdash (\E{k}\psi )\sigma \iff \mathcal{\tilde M},w \Vdash \E{k}\psi.$
By definition of the substitution we prove that:
$\mathcal{M},w \Vdash \E{k\sigma} \psi\sigma \iff \mathcal{\tilde M},w \Vdash \E{k}\psi.$

By definition of the realizability relation the condition is equivalent to:
$$\exists w' \in W: {\mathcal M}, w' \Vdash \psi\sigma \land w R_{k\sigma} w' \iff \exists w'' \in W: \mathcal{\tilde M}, w'' \Vdash \psi \land w \tilde R_k w''.$$
 By setting $w'=w''$, the following equivalence holds from (\ref{p2}): $w R_{k \sigma} w' \iff w \tilde R_{k} w'$. By induction hypothesis, we have: ${\mathcal M}, w' \Vdash \psi\sigma \iff \mathcal{\tilde M}, w' \Vdash \psi$.
The both last properties imply that:
$$\exists w' \in W: {\mathcal M}, w' \Vdash \psi\sigma \land w R_{k\sigma} w' \iff \mathcal{\tilde M}, w' \Vdash \psi \land w \tilde R_k w'$$ 
which implies the initial property.

\item $\mathcal{M},w \Vdash (\A{k}\psi )\sigma \iff \mathcal{\tilde M},w \Vdash \A{k}\psi.$ The proof is similar to the previous item.
\item $\mathcal M \Vdash (\textbf{E} \psi) \sigma \iff \mathcal {\tilde M} \Vdash \textbf {E} \psi.$
By definition of the substitution we prove that: 
$\mathcal M,w \Vdash \textbf {E} (\psi \sigma) \iff \mathcal {\tilde M},w \Vdash \textbf {E} \psi.$

By definition of the realizability relation, we have: 
$$\exists w \in W: \mathcal M,w \Vdash (\psi \sigma) \iff \mathcal {\tilde M},w \Vdash \psi,$$
which is directly verified by induction hypothesis.
\item $\mathcal M \Vdash (\textbf{A} \psi) \sigma \iff \mathcal {\tilde M} \Vdash \textbf{A} \psi.$ 
The proof is similar to the previous item.
\end{itemize}
\end{proof}
\end{proposition}

\begin{proof}[Proposition~\ref{prop:subst-behinc}]
First, let us remark that when $P \nSqsubset Q$, the property is trivially verified.
Besides, under the assumption $P \Sqsubset Q$, if $Q[\sigma]$ is not observable the property is also verified because an unobservable program includes all programs behaviourally (Definition~\ref{def:behinc}).

In the rest of the proof, we assume that $P$ is behaviourally included in $Q$ and $Q[\sigma]$ is observable (\ie $P \Sqsubset Q$ and $\obs Q[\sigma]$).
Hence, by definition of the observability there exists a model $\mathcal M$ such that $\mathcal M \Vdash \sem{Q[\sigma]}$. By proposition~\ref{prop:correct1}, we deduce that there exists a model $\mathcal{\tilde M}$ such that: $\mathcal{\tilde M} \Vdash \sem{Q}$. Moreover, as $P \Sqsubset Q$ by hypothesis, there exists $\tilde S \subseteq \dom \mathcal{\tilde M}$ such that: $\mathcal{\tilde M}_{\tilde S} \Vdash \sem{P}$. By construction of $\mathcal{\tilde M}$ we deduce that there exists a sub model of $\mathcal M$, denoted by $\mathcal M'$, complying to the properties, (\ref{p1}) and (\ref{p2}) of Proposition~\ref{prop:correct1} which corresponds to $\mathcal{\tilde M}_{\tilde S}$. Moreover, we have $\mathcal M' \Vdash P[\sigma]$ by Proposition~\ref{prop:correct1}. Then we conclude that: $ P[\sigma] \Sqsubset Q[\sigma]$.
\end{proof}

\begin{proof}[Theorem~\ref{the:fsr}]
First, let us remark that $P \Sqsubset Q$ is true whenever $\mathcal M \nVdash Q$ by definition of the behavioural inclusion (Definition~\ref{def:behinc}).
Hence, the proof doesn't consider the trivial verified case but rather the case where $\mathcal M \Vdash Q$.
\begin{description}

\item[Inst.] Directly from the definition of the behavioural inclusion (Definition~\ref{def:behinc}).

\item[Com.] By definition of the semantics $\sem{P,P'}=\textbf{A}( \phi \land \phi') = \textbf{A}( \phi' \land \phi) =
\sem{P',P}$ with $\sem{P}_P= \phi$ and $\sem{P'}_P= \phi'$. Thus, for all $\mathcal M$ we have: $\mathcal M \Vdash \sem{P,P'} \iff \mathcal M \Vdash \sem{P',P}$. 
Hence, if $Q \Sqsubset P,P'$ we conclude that: $Q \Sqsubset P',P$.

\item[Cont.] Similar to the proof of (Com.). 

\item[Asm.] First let us remark that $\restrict{\sigma}{\text{VA}(P) \cap \text{VA}(P')} = \restrict{\sigma'}{\text{VA}(P) \cap \text{VA}(P')}$ means that the substitution of the common variables are the same for $\sigma$ and $\sigma'$, leading to, 
$Q[\sigma \cup \sigma']=Q[\sigma]$ and $Q'[\sigma \cup \sigma']=Q'[\sigma']$. Let $\sigma'' = \sigma \cup \sigma'$. 
Then, we have the following property by definition of the semantics (Table~\ref{tab:semantics}) and $\sigma''$.
$$ \forall \mathcal M \in \textsf{KS}(\sem{(Q,Q')[\sigma'']}): \mathcal M \Vdash \sem{Q[\sigma]} \land \mathcal M \Vdash \sem{Q'[\sigma']}.$$
Notice that the set of models, $\textsf{KS}(\sem{(Q,Q')[\sigma'']})$, is not empty since, by hypothesis,\\ $\obs (Q[\sigma],Q'[\sigma'])$ holds.
As $Q \leftfree_\sigma P$ and $Q' \leftfree_{\sigma'} P'$, any model validating $Q$ (resp. $Q'$) also validates $P$, (resp. $P'$) by definition of the functional synthesis. Then, we deduce that: 
$$\forall \mathcal M \in \textsf{KS}(\sem{(Q,Q')[\sigma'']}): \mathcal M \Vdash \sem{P[\sigma]} \land \mathcal M \Vdash \sem{P'[\sigma']}.$$
Then, we conclude that:
$$\forall \mathcal M \in \textsf{KS}(\sem{(Q,Q')[\sigma'']}): \mathcal M \Vdash \sem{(P,P')[\sigma'']}.$$
\end{description}
\end{proof}
\subsection*{Complete compilation of the Band Detector}
\label{sec:proofEx}

\begin{table}[ht]
\begin{sideways}
\begin{tabular}{c}
\textsc{- Sender -} \\
{\tiny
\begin{array}[t]{c c c}
\multicolumn{3}{c}{
\AxiomC{$Q_1,Q_2,Q_3[\sigma=\{detect/Detect,light/Light,v_1/Tetr, v_2/Luxl\}] \subseteq_{Asm} P_{11}'[\sigma]$}
\AxiomC{$obs(Q_1,Q_2,Q_3[\sigma])$}
\RightLabel{(Inst.)}
\BinaryInfC{$Q_1,Q_2,Q_3\leftfree_{\sigma} P_{11}'\}$}
\DisplayProof}
\\
&&\\
\multicolumn{3}{c}{
\AxiomC{$Q_1,Q_2,Q_3[\sigma'=\{detect/Detect,light/Light,v_3/Tetr, v_4/Luxl\}] \subseteq_{Asm} P_{12}'[\sigma']$}
\AxiomC{$obs(Q_1,Q_2,Q_3[\sigma'])$}
\RightLabel{(Inst.)}
\BinaryInfC{$Q_1,Q_2,Q_3\leftfree_{\sigma'} P_{12}'$}
\DisplayProof}
\\
&&\\
\multicolumn{3}{c}{
\AxiomC{$Q_1,Q_2,Q_3[\sigma''=\{detect/Detect,light/Light,v_5/Tetr, v_6/Luxl\}] \subseteq_{Asm} P_{13}'[\sigma'']$}
\AxiomC{$obs(Q_1,Q_2,Q_3[\sigma''])$}
\RightLabel{(Inst.)}
\BinaryInfC{$Q_1,Q_2,Q_3\leftfree_{\sigma''}P_{13}'$}
\DisplayProof}
\\
&&\\
\AxiomC{$P_{11}'=\context{light}\{detect\pcause v_1, v_1\pcause v_2, v_2 \pcause AHL(low)\}$}
\RightLabel{(Trans.)}
\UnaryInfC{$\context{light}\{detect\pcause v_1, v_1\pcause AHL(low)\}$}
\RightLabel{(Trans.)}
\UnaryInfC{$\context{light}\{detect\pcause AHL(low)\}$}
\RightLabel{(N2P.)}
\UnaryInfC{$P_{11}$}
\DisplayProof
&
\AxiomC{$P_{12}'=\context{light}\{detect\pcause v_3, v_3\pcause v_4, v_4 \pcause AHL(mid)\}$}
\RightLabel{(Trans.)}
\UnaryInfC{$\context{light}\{detect\pcause v_3, v_3\pcause AHL(mid)\}$}
\RightLabel{(Trans.)}
\UnaryInfC{$\context{light}\{detect\pcause AHL(mid)\}$}
\RightLabel{(N2P.)}
\UnaryInfC{$P_{12}$}
\DisplayProof
&
\AxiomC{$P_{13}'=\context{light}\{detect\pcause v_5, v_5\pcause v_6, v_6 \pcause AHL(high)\}$}
\RightLabel{(Trans.)}
\UnaryInfC{$\context{light}\{detect\pcause v_5, v_5\pcause AHL(high)\}$}
\RightLabel{(Trans.)}
\UnaryInfC{$\context{light}\{detect\pcause AHL(high)\}$}
\RightLabel{(N2P.)}
\UnaryInfC{$P_{13}$}
\DisplayProof
\\ 
&& \\
\multicolumn{3}{c}{
\AxiomC{$Q_1,Q_2,Q_3\leftfree_{\sigma} P_{11}$}
\AxiomC{$Q_1,Q_2,Q_3\leftfree_{\sigma'} P_{12}$}
\AxiomC{$Q_1,Q_2,Q_3\leftfree_{\sigma''} P_{13}$}
\RightLabel{(Asm.)}
\TrinaryInfC{$Q_1,Q_2,Q_3\leftfree_{\sigma\cup{\sigma}'\cup{\sigma}''} P_{11},P_{12},P_{13}$}
\DisplayProof
}
\end{array} 
} 
\\ \\
\textsc{- Receiver -} \\
{ \tiny
\begin{array}[t]{c c c}
\multicolumn{3}{c}{
\AxiomC{$Q_4,Q_5,Q_6,Q_8[\sigma=\{v_1/LuxR,v_2/Cl,v_3/Lacl\}] \subseteq_{Asm} P_{21}'[\sigma]$}
\AxiomC{$obs(Q_4,Q_5,Q_6,Q_8[\sigma])$}
\RightLabel{(Inst.)}
\BinaryInfC{$Q_4,Q_5,Q_6,Q_8\leftfree_\sigma P_{21}'$}
\DisplayProof}\\
&&\\
\multicolumn{3}{c}{
\AxiomC{$Q_4,Q_5,Q_6,Q_8[\sigma'=\{v_4/LuxR,v_5/Cl,v_6/Lacl\}] \subseteq_{Asm} P_{22}'[\sigma']$}
\AxiomC{$obs(Q_4,Q_5,Q_6,Q_8[\sigma'])$}
\RightLabel{(Inst.)}
\BinaryInfC{$Q_4,Q_5,Q_6,Q_8\leftfree_\sigma' P_{22}'$}
\DisplayProof}\\
&&\\
\multicolumn{3}{c}{
\AxiomC{$Q_4,Q_5,Q_7[\sigma''=\{v_7/LuxR,v_8/LacM1\}] \subseteq_{Asm} P_{23}'[\sigma'']$}
\AxiomC{$obs(Q_4,Q_5,Q_7[\sigma''])$}
\RightLabel{(Inst.)}
\BinaryInfC{$Q_4,Q_5,Q_7\leftfree_\sigma'' P_{23}'$}
\DisplayProof}\\
&&\\
\AxiomC{$P_{21}'=AHL(low)\pcause v_1,v_1\pcause v_2,v_2\pcause v_3,v_3\pcause \overline{GFP}$}
\RightLabel{(Trans.)}
\UnaryInfC{$AHL(low)\pcause v_1,v_1\pcause v_2,v_2\pcause \overline{GFP}$}
\RightLabel{(Trans.)}
\UnaryInfC{$AHL(low)\pcause v_1,v_1\pcause \overline{GFP}$}
\RightLabel{(Trans.)}
\UnaryInfC{$AHL(low)\pcause \overline{GFP}$}
\RightLabel{(N2P.)}
\UnaryInfC{$P_{21}$}
\DisplayProof
&
\AxiomC{$P_{22}'=AHL(mid)\pcause v_4,v_4\pcause v_5,v_5\pcause v_6,v_6\pcause GFP$}
\RightLabel{(Trans.)}
\UnaryInfC{$AHL(mid)\pcause v_4,v_4\pcause v_5,v_5\pcause GFP$}
\RightLabel{(Trans.)}
\UnaryInfC{$AHL(mid)\pcause v_4,v_4\pcause GFP$}
\RightLabel{(Trans.)}
\UnaryInfC{$AHL(mid)\pcause GFP$}
\RightLabel{(N2P.)}
\UnaryInfC{$P_{22}$}
\DisplayProof
&
\AxiomC{$P_{23}'=AHL(high)\pcause v_7,v_7\pcause v_8,v_8\pcause \overline{GFP}$}
\RightLabel{(Trans.)}
\UnaryInfC{$AHL(high)\pcause v_7,v_7\pcause \overline{GFP}$}
\RightLabel{(Trans.)}
\UnaryInfC{$AHL(high)\pcause \overline{GFP}$}
\RightLabel{(N2P.)}
\UnaryInfC{$P_{23}$}
\DisplayProof
\\
&&\\
\multicolumn{3}{c}{
\AxiomC{$Q_4,Q_5,Q_6,Q_8\leftfree_\sigma P_{21}$}
\AxiomC{$Q_4,Q_5,Q_6,Q_8\leftfree_{\sigma'} P_{22}$}
\AxiomC{$Q_4,Q_5,Q_7\leftfree_{\sigma''} P_{23}$}
\RightLabel{(Asm.)}
\TrinaryInfC{$Q_4,Q_5,Q_6,Q_7,Q_8\leftfree_{\sigma\cup{\sigma}'\cup{\sigma}''} P_{21},P_{22},P_{23}$}
\DisplayProof
}
\end{array}
} 
\\ \\
\textsc{- Final design -}
\\
{\tiny
\footnotesize
$ \begin{array}{|l|l|}
\hline
\multicolumn{1}{|c|}{\mathit{Sender}}& \multicolumn{1}{c|}{\mathit{Receiver}}\\
\hline
\{\textsf{AHL:}\{low \napprox mid \napprox high \}, \qquad \context{\text{Light}}\{detect\cause\textsf{Tetr}\},
&\{\textsf{AHL:}\{low \napprox mid \napprox high \}, \qquad
\textsf{LuxR:}\{low \napprox \{mid \prec high\} \},\\
\textsf{Tetr}\activate\textsf{LuxL}, \hfill
\textsf{Luxl}\activate\textsf{AHL}(low),
& \textsf{AHL}(mid)\cause\textsf{LuxR}(mid),\hfill
 \textsf{AHL}(high)\cause\textsf{LuxR}(high),\\
\textsf{Luxl}\activate\textsf{AHL}(mid), \hfill
\textsf{Luxl}\activate\textsf{AHL}(high)\}
& \textsf{LuxR}(mid)\activate\textsf{Cl}, \hfill
\textsf{LuxR}(high)\activate\textsf{LaclM1},\\
& \textsf{Cl}\inhibit\textsf{Lacl},\qquad
\textsf{LaclM1}\inhibit\textsf{GFP}, \hfill
 \textsf{Lacl}\inhibit\textsf{GFP}\}\\
\hline
\end{array} $
}
\end{tabular}
\end{sideways} 
\vspace{-7ex}\caption{Complete band detector compilation.}
\label{tab:proof2}
\end{table}
\end{document}